\documentclass[conference]{IEEEtran}
\usepackage[utf8]{inputenc}
\usepackage[T1]{fontenc}
\usepackage{url}
\usepackage{graphicx}
\usepackage{amsfonts}
\usepackage{amsmath}
\usepackage{amsthm}
\usepackage{subfig}
\usepackage{enumerate}
\usepackage{multirow}
\usepackage{eqparbox}
\usepackage{epstopdf}
\usepackage[noadjust]{cite}
\usepackage{color}
\usepackage{todonotes}
\usepackage{booktabs}
\usepackage[normalem]{ulem}

\newcommand{\found}{\emptyset}
\newcommand{\kad}[5]{
$\mathcal{K}(#1,#2,#3,#4,#5)$}
\newcommand{\calO}{\mathcal{O}}
\newcommand{\Hop}{H}

\newtheorem{theorem}{Theorem}[section]
\newtheorem{lemma}[theorem]{Lemma}

\newtheorem{definition}[theorem]{Definition}
\title{Comprehending Kademlia Routing - A Theoretical Framework for the Hop Count Distribution}
\author{\IEEEauthorblockN{Stefanie Roos and Hani Salah and Thorsten Strufe}
\IEEEauthorblockA{
Peer-to-Peer Networks Group\\
Technische Universit\"at Darmstadt \\
\{roos, hsalah, strufe\}@cs.tu-darmstadt.de}}

\newcommand{\schange}[2]{#2}

\begin{document}
\maketitle

\begin{abstract}

The family of Kademlia-type systems represents the most efficient and most widely deployed class of internet-scale distributed systems. Its success has caused plenty of large scale measurements and simulation studies, and several improvements have been introduced. Its character of parallel and non-deterministic lookups, however, so far has prevented any concise formal analysis.
This paper introduces the first comprehensive formal model of the routing of the entire family of systems that is validated against previous measurements. It sheds light on the overall hop distribution and lookup delays of the different variations of the original protocol. It additionally shows that several of the recent improvements to the protocol in fact have been counter-productive and identifies preferable designs with regard to routing overhead and resilience.
%
\end{abstract}

\section{Introduction}
\label{sec:intro}

Distributed Hash Tables (DHTs) received considerable attention during the last decade. 
On an abstract level, DHTs allow the mapping of objects to nodes in a completely
decentralized and highly dynamic network on the basis of IDs, such that
both the number of nodes contacted during object retrieval and the connections
maintained by each node increase logarithmically with the network size.
As a consequence, DHTs are candidate solutions for large-scale distributed
data storage as well as for decentralized resilient communication systems.

In practice, only variants of the Kademlia DHT \cite{Maymounkov02Kademlia}
have been deployed successfully, attracting millions of users in the file-
sharing applications BitTorrent and eMule \cite{junemann11towards, salah13capturing}. 
Even in networks of such an enormous size, the discovered routes are generally in the order of $3$
to $4$ hops \cite{Stutzbach06improving,Steiner10eval,Jimenez2011subsecond}.
Additionally, Kademlia's redundant routing tables combined with an iterative parallel lookup scheme 
make it particularly suitable for dynamic environments.

Despite the considerable attention Kademlia received from both research and industry,
the impact of the design parameters on the routing performance is poorly understood.
Measurements only offer insights on deployed systems, whereas simulations do not scale
beyond several ten thousands of nodes.

In order to assess different design choices, a concise model for the
complete hop count distribution 
is needed, which covers all the existing Kademlia implementations as well as 
allowing for a huge variety of modifications.
The model is required to give a close bound on the hop count distribution based on the routing table structure and the routing algorithm, and it has to consider the impact of churn and routing table incompleteness, while
still being computationally efficient.

We model routing in Kademlia as a Markov chain with a multi-dimensional
state space (Section \ref{sec:model}). Our derivation provides extremely tight upper and lower bounds
on the hop count distribution, and covers a wide range of overlay topologies.
Interestingly, analyzing the topologies of deployed systems, we find that they do not always outperform the original protocol. 
Furthermore, the analysis of the parameter space
enables us to derive guidelines for design decisions (Section \ref{sec:results}).

The computation of the hop count is efficient, requiring $\calO(n polylog(n))$ basic operations and $\calO(polylog(n))$ storage space.
Given such a moderate increase, networks of up to $1B$ nodes can easily be analyzed.

The model is verified in two ways: 
First, the initial model for a static environment is verified by simulations (Section \ref{sec:verification}). Secondly, the extended model, allowing for churn and routing table incompleteness, is compared
to  measurements made by Stutzbach and Rejaie for the KAD network \cite{Stutzbach06improving}, resulting in an error rate of $2.7 \%$ for the average hop count, 
in contrast to $5.5 \%$  provided by their analytic model (Section \ref{sec:churn}).

\section{Kademlia-type Systems}
\label{sec:related}


Kademlia is a structured peer-to-peer (P2P) system \cite{Maymounkov02Kademlia}. 
Nodes and objects are assigned IDs from the same $b$-bit identifier space and the distance between two identifiers is defined as the XOR of their values.
Kademlia implements key-based routing and storage of key-value (identifier-object) pairs. The nodes at the closest distance to an object's identifier are responsible for storing it.

Each node $v$ maintains a routing table to store the ID and address of other nodes, without keeping persistent network connections to them.
Known nodes are also called \emph{contacts}, and in case that contacts leave the system, the information stored about them may become outdated, which is commonly termed \emph{stale}.
The routing table in Kademlia is structured as a tree, which consists of a $k$-bucket storing up to $k$ contacts at each of its level $i$ (with $i$ $\in$ $[0, b)$ being the common prefix length of a contact and $v$).


Kademlia implements greedy routing: To route a message from node $v$ to a target identifier $t$ (for the storage or retrieval of objects), $v$ sends parallel lookup requests to the $\alpha$ known contacts that are closest to $t$. Every queried contact that is online replies with the set of $\beta$ nodes that are locally known as being closest to $t$, thus extending $v$'s set of candidate contacts. This process is iterated until the lookup does not produce any contacts closer to $t$ than previously have been discovered, or a timeout is caught.
The original Kademlia publication suggests to use $\alpha=3$ and $k=20$.

Kademlia proved highly efficient and reliable, and thus has frequently been modified, generating a broad family of Kademlia-type systems.
Each adaptation mainly adjusts the given parameters, or the routing table structure.

The current mainline implementation of BitTorrent (MDHT), for example, integrates a Kademlia-type DHT for node discovery.
uTorrent, the most popular client implementing MDHT, is implemented using $8$-buckets, $\alpha=4$, and $\beta=1$ \cite{Jimenez2011subsecond}.

To reflect the fractions of the namespace that is covered at different levels, and hence to increase the distance reduction at each hop, \cite{Jimenez2011subsecond} introduces variable bucket sizes $k_i$ (iMDHT).
They are chosen to be 128, 64, 32, and 16 for the buckets at levels $i \in (0..3)$ respectively, and 8 for all lower levels.

The Kademlia-type DHT used in the highly popular eMule file-sharing system, KAD, adds multiple buckets per level, grouping contacts according to the first $l$ bits after the first non-common bit. This way, the \emph{bit gain}, i.e. the difference between the common prefix length of the current hop and the next hop to $t$, is at least $l$. Choosing $k$ to be 10, the implementation contains buckets for all 4-bit prefixes at the lowest level (including contacts that share no common prefix with $v$), and one bucket for each of the sub-prefixes $111$, $110$, $101$, $1001$, and $1000$. Thus, the guaranteed distance reduction is $3$ bits for $75\%$ of the targets IDs, and $4$ bits for the remaining $25\%$. By default, KAD implements $\alpha=3$ and $\beta=2$.

The success of Kademlia-type systems has caused a large number of studies over the past few years \cite{Steiner07global, Wang08attacking, Steiner10eval, falkner07profiling, Jimenez2011subsecond, Crosby07ananalysis}, which are mainly based on large-scale measurements but do not yield insight into the impact of isolated design adaptations.
%
 
Analytic results of the routing are rare.
Existing studies are largely restricted to the asymptotic worst-case complexity of $\calO(\log n)$ routing steps for a network of order $n$.
A notable exception is \cite{Stutzbach06improving}, in which 
a formula for the average hop count is derived. 
This derivation, however, considers only the KAD implementation and fails to give further insight into the hop count distribution.
It hence does not allow for the choice of sensible timeout durations and termination criteria for more sophisticated, possibly time-critical applications.
Parallel lookups and a non-constant bucket size furthermore are disregarded.

In this paper, we show a derivation of the complete hop count distribution, which does not only
cover all the deployed Kademlia-type systems, but also allows a straightforward analysis of new designs.
\section{Model}
\label{sec:model}
%
The \emph{hop count} refers to the number of edges on the shortest path that has been traversed during the lookup process. Each routing step (i.e. \emph{hop}) is a \emph{transition} from a set of queried contacts to either another set of queried contacts or \schange{a successful}{routing} termination. 

In our model, a \emph{state} is defined by the common prefix lengths of the currently-known closest contacts with the target. That is, the state space of the Markov chain consists of $\alpha$-dimensional integer vectors. The \emph{initial distribution} $I$ corresponds to the common prefix lengths of the closest $\alpha$ contacts in the requester's routing table. 
\schange{A transition between a state $x=(a_1, \ldots , a_\alpha)$ and 
a state $y=(c_1, \ldots, c_\alpha)$
corresponds to the event that after querying $\alpha$ contacts
with common prefix lengths $a_1, \ldots , a_\alpha$,   
the new closest $\alpha$ contacts have common prefix lengths 
$c_1, \ldots, c_\alpha$.}
{A hop in the routing corresponds to a \emph{transition} from one $\alpha$-dimensional vector of common prefix lengths
to a either a second vector of common prefix lengths or routing termination.}

With the \emph{initial distribution} $I$ and the \emph{transition matrix}
$T$, the common prefix length distribution of the nodes queried in the $i$-th step is
$\Hop_i = T^{i-1}I$. 
As a consequence, the \emph{cumulative hop count distribution} can be obtained from $\Hop_i$
as the fraction of queries that have reached the terminal state. 
\schange{}{Due to the Markov property, our model fails to cover the improbable,
but technically possible, event that nodes other than those returned
from the most recent query are chosen to be contacted because 
the most recent query has not provided $\alpha$ distinct closer nodes.
We overcome this insufficiency by computing $T^{up}$ and $T^{low}$,
which provide an upper and lower bound, respectively, on the
fraction of terminated queries.
In the following, we first derive the distribution of closest entries in a
routing table in Section \ref{sec:closest}, which allows us to derive
$I$ in Section \ref{secI} and $T$ in Section \ref{sec:T}.

}

\schange{
The main difficulty in determining $I$ and $T$ is to derive the common prefix length 
distribution of the $\gamma \in \{\alpha, \beta \}$
closest contacts in a node $v$'s routing table.
Note that $v$ can be either the requester (initial distribution)
or a queried contact (transitions).
The \emph{closest contact distribution} is obtained in two steps. First, the probability
of contacting the target is computed as the ratio of the bucket size and
all nodes within the fraction of the identifier space corresponding to that bucket. Due to uniformly selected IDs, the latter
is modeled as a binomially distributed random variable. The common prefix lengths of all $\alpha$ new contacts are of interest, only if the target is not found.
Then the common prefix length distribution is obtained as the maximum $\gamma$ values
in a set of $k$ identically distributed random variables, which
represent the IDs in a bucket, chosen uniformly at random from all
IDs in the respective fraction of the ID space. Having derived the closest contact distribution, 
the initial distribution follows directly.
The \emph{transition probabilities} require an additional step of deriving
a distribution over the closest $\alpha$ distinct candidates after
obtaining distributions over up to $\alpha \beta$ returned
contacts.}
{}

\subsection{Assumptions}
\label{sec:assumptions}
We model a query for an ID of an existing node\schange{ rather than
a query for a file with potentially a high number of replicas.
But the latter case is very similar, and can be achieved
by slight changes.}{.}
Our basic model relies on the following assumptions, which allow 
a very general, application-independent view.
\schange{}{We first state the assumptions, before elaborating on their motivation
and impact on the validity of the model.}

\begin{enumerate}
\item There are no stale contacts in the routing tables.
\item Nodes do not fail nor do they drop messages.
\item Buckets are maximally full, i.e. if a bucket contains $k_1 < k$ values,
there are exactly those $k_1$ nodes in the region the bucket is responsible for.
\item Node IDs are uniformly and independently distributed over the whole identifier space.
\item Routing table entries are chosen independently. 
\item If the distance between a node and the target ID is $0$, the node's routing table contains the target.
\item The lookup uses strict parallelism, i.e. a node awaits all answers to its queries before 
sending additional ones.
\end{enumerate}

Assumptions 1, 2, and 3 can be summarized as the assumption of a steady-state system, without churn
or failures.
However, we extend the model in Section \ref{sec:churn} to allow
for churn and bucket incompleteness. 
Note that for applications such as critical infrastructures
and large data centers, churn is basically
non-existent and the failure rate can assumed to be
negligible. 
Assumption 4 is given in general, since the ID is usually chosen randomly
or as a hash of some identifying value, e.g. the IP address.
\schange{
Assumption 5 is slightly problematic, since nodes in a high number
of routing tables are more likely to be chosen by new contacts,
but our comparison with real-world measurements (see Section \ref{sec:churn}) indicates that this assumption does not dramatically
influence the quality of the model. 
}
{Assumption 5 holds in as far as that nodes discover contacts initially by
searching for their own ID which should result with contacts close to their
own ID independently of the starting point.
However, nodes encountered and potentially added during routing tend to have
a higher than average in-degree.
By this, the chance that a node in one routing table is present in another 
is slightly higher than the chance that a random node is contained in
a routing table.
Still, the probability should be negligible for large networks,
as indicated from the agreement of our model with real-world measurements.
}
Assumption 6 considers the case that multiple nodes share the same ID,
which can happen by Assumption 4 (independent choice), but is highly
unlikely in practice and hence only a theoretical construct to simplify
the derivations.
Assumption 7
is consistent with various implementations, whereas
others allow interleaving queries as well as more than $\alpha$
concurrently out-standing answers. Steiner et al. present an analysis of how the
latency can be enhanced by immediately reacting to a query \cite{Steiner08faster}.
However, it is not possible to always select the closest of all returned contacts, 
resulting in a higher hop count and number of contacted nodes.
Consequently, strict parallelism is optimal with regard to our metric of interest,
the hop count.

\subsection{Model Overview}
\label{sec:formal}
In the following, the idea is formalized, defining 
the parameters governing the routing and the states of the Markov chain.
   
The common prefix
lengths of the closest nodes to the target ID $t$ is used to characterize
the routing process.
Because routing is commonly modeled as a monotonously decreasing process
that converges to a distance of $0$, we define the distance of
two nodes $w$ and $v$ to be
\begin{align}
\label{eq:dist}
dist(w,v) &= b - commonprefixlength(id(w), id(v)) \nonumber \\
&= \lfloor \log_2 XOR(id(w) ,id(v)) \rfloor +1,
\end{align}  
where $b$ is the ID space size and $id(v)$ denotes the $b$-bit
ID of node $v$. 
\schange{In the following, distance refers to $dist$ as defined
above, unless stated otherwise.}
{We here use distance to refer to $dist$ rather than the XOR distance, unless stated otherwise.}

The state of a query is either $\found$, denoting a terminated query,
or the distance of the currently queried $\alpha$ nodes to the target $t$.
Formally, the state space is given by
\begin{align}
\label{eq:state}
&S_\alpha = \{\found\} \cup S'_\alpha \textnormal{ with } \\
S'_\alpha = \{ x=(x_1,\ldots, x_\alpha)& \in \mathbb{Z}^\alpha_{b+1}: x_j \leq x_{j+1}, j=1\ldots \alpha-1 \}. \nonumber
\end{align} 
Aiming to reduce the number of states and consequently the storage and computation cost, we assume the vector of distances to be sorted in ascending order.

It remains to define the parameters influencing the hop count.
We characterize a Kademlia-type system
by the ID space size $b$, the routing parameter $\alpha$ and $\beta$,
as well as the routing table parameters $k$ and routing table structure $L$
\schange{.}{, which determines the number of buckets per level as well
as how the ID space is split between those buckets.}
\begin{definition}
A  \kad{b}{\alpha}{\beta}{k}{L}-system is a Kademlia-type system
with the following properties:
\begin{itemize}
\item A $b$-bit identifier space is used for addressing.
\item $\alpha$ parallel iterative queries are send for each lookup.
\item Each queried node answers with at most $\beta$ contacts closer to the target than itself.
\item The $d$-entry $k_d$ of the vector $k \in \mathbb{N}_0^{b+1}$ gives the bucket size for nodes with distance $d$ 
to the routing table owner (i.e. the bucket size at level $b-d$).
\item \schange{The entry $L_{ij}$ of the matrix $L \in \mathbb{R}^{(b+1) \times (b+1)}$ gives the fraction of the ID space at
level $b-i$ for which $j$ more bits are resolved}
{ The $i$-th row of the matrix $L$ gives the distribution of the guaranteed bit gain at distance $i$ to the routing table owner, i.e. the entry $L_{ij}=\frac{x}{2^{i}}$ is defined by the number $x$
of IDs with distance $i$  that are sorted in buckets covering a region of $2^{i-j}$ IDs each.}
\end{itemize}
\end{definition}
Furthermore, the network order $n$ influences the hop count distribution. 
Note that in most Kademlia-type systems, such
as MDHT as well as KAD, $k$ is constant.
Similarly, the matrix $L$ is commonly sparse.
For instance, in MDHT only one bucket is used for each common prefix length,
so $L_{i1}=1$ for $i=0\ldots b$ and $L_{ij}=0$ in all other cases.
KAD is more complicated: 
$L^{KAD}_{b4}=1$, $L^{KAD}_{i3} = 0.75$, and $L^{KAD}_{i4} = 0.25$
for $i < b$ determine the routing table structure in the KAD system.
This is due to resolving at least $4$ more bits on the top level,
and splitting into buckets with prefixes $111$, $110$, $101$ ($75$\% of IDs), 
as well as $1001$ and $1000$ ($25$\% of IDs) for
all lower levels. 

This completes the introduction of the basic terminology.  
Next, the distribution of closest contacts in a node's routing
table is obtained, which is essential for computing both
the initial distribution as well as the transition matrix.
\subsection{Distribution of closest contacts}
\label{sec:closest}
We are interested in the distribution of the closest $\gamma \in \{ \alpha, \beta \}$ 
contacts to a target $t$ in a routing table of a node $v$.
Let the random variable $X_0$ with values in $\mathbb{Z}_{b+1}$ 
be the distance of $v$ to $t$. The random variable $X_1$ with values in $S_\gamma$ gives the state
characterizing the closest neighbors. 
In the following, we derive the probability $P(X_1 = s | X_0 = d)$.

By Assumption 6, the case $d=0$ is trivially given by
\begin{align}
\label{eq:d0}
&P(X_1 = s | X_0 = 0)
 = \begin{cases} 
 1, \quad s=\found \\
 0, \quad s\neq \found
 \end{cases}
 .
\end{align}
So, from now  we consider $d > 0$ 
\schange{.Additionally, we condition on the number of further
resolved bits $L_d=l$.
The general result is then obtained by computing the expectation over the $d$-th row
of the matrix $L$. }{. The success probability is determined by the 
distribution for the guaranteed bit gain $L_d$ defined by the $d$-th row of the matrix $L$
and the additional bit gain dependent on the bucket size $k_d$. 
\begin{align}
\label{eq:T1total}
\begin{split}
&P(X_1 = s | X_0 = d) \\
&= \sum_{l=0}^bP(X_1=s|X_0=d,L_d=l)P(L_d=l) \\
&= \sum_{l=0}^b P(X_1=s|X_0=d,L_d=l)L_{dl}
\end{split}
\end{align}
It remains to obtain $P(X_1=s|X_0=d,L_d=l)$.
}

We start by determining the probability to reach the state $\found$.
Recall that $r$'s routing table has $k_d$-buckets of nodes that differ in the $b$-$d$-th bit.
Let $x$ be the number of candidate nodes to be in the bucket, i.e. the
number of nodes in the respective part of the ID space.
If the bucket contains less than $k_d$ contacts, by \schange{our assumption of maximally full buckets (Assumption 3)}
{Assumption 3},
$t$ is one of them with probability $q_m=1$. 
Otherwise, with $m$ candidates, $t$ is contained in the
bucket with probability 
$q_m = \frac{k_d}{m+1}$\schange{}{, the likelihood that $t$ is one of the $k_d$ nodes selected 
among $m+1$ nodes}.
If a node has distance at most $d$-$l$ to $t$, there are $2^{d-l}$ IDs it
can potentially have,
making up a fraction $\frac{2^{d-l}}{2^b} = 2^{d-l-b}$ of all IDs.
The number of nodes $X$ within a fraction $2^{d-l-b}$ of the ID space is binomially distributed, 
$X \sim B(n-2,2^{d-l-b})$ ($n-2$ because $t$ and $v$ are excluded), by Assumption 4. 
So the probability that $t$ is contained in the routing table is computed as
\begin{align}
\label{eq:F1}
&P(X_1 = \found | X_0 = d,B_d=l) 
= \sum_{m=0}^{n-2} P(X=m)q_m \nonumber \\
&= \sum_{m=0}^{k_d} {n-2 \choose m} \left(2^{d-l-b}\right)^m \left(1-2^{d-l-b}\right)^{n-2-m} \\
&+ \sum_{m=k_d+1}^{n-2} {n-2 \choose m} \left(2^{d-l-b}\right)^m \left(1-2^{d-l-b}\right)^{n-2-m} \frac{k_d}{m+1}. \nonumber
\end{align}
If, on the other hand, $t$ is not contained in the routing table, we need to derive
$P(X_1 = (\delta_1,...,\delta_\gamma) | X_0= d,L_d=l)$
for all $(\delta_1,\ldots,\delta_\gamma) \in S_\gamma$.
The distribution of distances within one bucket is needed. 
The probability that a contact has a certain distance corresponds to the fraction of IDs with this distance.
Consequently, the cumulative distribution function $F_{d,l}$ of the distance of one randomly chosen
contact in a bucket of contacts with distance at most $d$-$l$ is given by
\begin{align} 
\label{eq:cdf}
F_{d,l}(x) = \min \{1, \frac{2^{\lfloor x \rfloor}}{2^{d-l}} \}.
\end{align}
for $x\geq 0$. 
Knowing the distance distribution of a random contact, one can derive the distance of
the $\gamma$ closest contacts.
First, we rewrite the vector $(\delta_1,...,\delta_\gamma)$, grouping identical values.
This later allows us to treat the number of nodes with the same distance as a binomially distributed random variable.
More specifically, the transformation $M$ is applied to $X_1$ in order to obtain tuples $M_1,\ldots , M_{\gamma '} \in \mathbb{Z}^2$,
so that the first component of $M_i = (y_i, c_i)$ is the $i$-th smallest value in $(\delta_1,...,\delta_\gamma)$
and $c_i$ is the number of occurrences of $y_i$ in $(\delta_1,...,\delta_\gamma)$. 
For reasons of presentation, we set $M_0 = (y_0,c_0)=(-1,0)$.
As a result, an equivalent expression
for the probability distribution of $X_1$ is the following:
\begin{align}
\label{eq:M}
&P(X_1 = (\delta_1,...,\delta_\gamma) | X_0= d,L_d=l) \nonumber \\
=&P(M(X_1) = ((y_1,c_1), \dots , (y_{\gamma'},c_{\gamma'})) | X_0= d,L_d=l) \nonumber \\
=&\big(1-P(X_1 = \found | X_0 = d,L_d=l)\big)\\
\cdot & \prod_{i=1}^{\gamma'} P\big(M_i = (y_i,c_i) | X_0= d,L_d=l, X_1 \neq \found,  \nonumber \\
& \quad \quad  \quad M_{0} = (y_0,c_0) , \ldots , M_{i-1} = (y_{i-1},c_{i-1}) \big) \nonumber
\end{align}
 

It remains to compute each factor in Eq. \ref{eq:M}.
We first treat the case $i < \gamma'$, for which 
we have to determine the probability that 1) all $C_{i-1}= k_d-\sum_{j=1}^{i-1} c_j$ bucket entries
\schange{}{with distance at least $y_{i-1}+1$} are
at distance at least $y_i$ to the target, and 2) there are \emph{exactly} $c_i$ such entries.
\schange{More precisely, the first event means that there are no contacts
at distance $y_{i-1}+1$ to $y_i-1$. }{}
More precisely, event 2) conditioned on event 1) corresponds to the event that a binomially distributed random variable with 
$C_{i-1}$ trials and success probability $p_i = \frac{F_{d,l}(y_i)-F_{d,l}(y_i-1)}{1- F_{d,l}(y_i-1)}$ has exactly $c_i$ successes. 
Note that the number of trials $C_i$ and the denominator $1- F_{d,l}(y_i-1)$ result from conditioning 
on $M_0, \ldots, M_{i-1}$ and event 1), respectively.
Using the above terminology, we get:   
\begin{align}
\label{eq:Mi}
&P\big(M_i = (y_i,c_i) | X_0= d,L_d=l, X_1 \neq \found, \nonumber \\
& \quad M_{0} = (y_0,c_0) , \ldots , M_{i-1} = (y_{i-1},c_{i-1})\big) \nonumber\\
=&P\big(M_i(1) \geq y_i | X_0= d,L_d=l, X_1 \neq \found, \nonumber \\
&\quad M_{0} = (y_0,c_0) , \ldots , M_{i-1} = (y_{i-1},c_{i-1})\big)\\
\cdot &P\big(M_i = (y_i,c_i) | X_0= d,L_d=l, X_1 \neq \found, \nonumber\\
& \quad M_{0} = (y_0,c_0) , \ldots , M_{i-1} = (y_{i-1},c_{i-1}),M_i(1) \geq y_i\big) \nonumber\\
=&\left(\frac{1 - F_{d,l}(y_i -1)}{1-F_{d,l}(y_{i-1})}\right)^{C_{i-1}}
{C_{i-1}\choose c_i} 
 p_i^{c_i} \left(1- p_i\right)^{C_i} \nonumber
\end{align}

For the $\gamma'$-th distinct value, the probability that there are \emph{at least} $c_{\gamma'}$  equal values rather than exactly $c_{\gamma'}$ values is derived. 
There might be other contacts with the same distance in the bucket, which are not part of 
\schange{the closest $\gamma$ contacts}{the chosen $\alpha$ contacts}. 
So, we have: 
\begin{align}
\label{eq:Mlast}
&P\big(M_{\gamma'} = (y_{\gamma'},c_{\gamma'}) | X_0= d,L_d=l, X_1 \neq \found, \nonumber \\
& \quad M_{1} = (y_1,c_1) , \ldots , M_{i-1} = (y_{\gamma'-1},c_{\gamma'-1})\big)\\
=&\left(\frac{1-F_{d,l}(y_{\gamma'})}{1-F_{d,l}(y_{\gamma'-1})}\right)^{C_{\gamma'-1}} \nonumber \\
&\left( 1 - \sum_{j=0}^{c_{\gamma'}-1} {C_{\gamma'-1}\choose j} p_{\gamma'}^{j}  \left(1- p_{\gamma'}\right)^{C_{\gamma'-1}-j}\right) \nonumber
\end{align}

\schange{Combining Eqs. \ref{eq:Mi}, and \ref{eq:Mlast}, the probability distribution for $X_1$ \schange{given fixed $d$ and $l$}{}
is defined by Eq. \ref{eq:M}. 
Note that
so that Eq. \ref{eq:M} can be simplified to}
{By Eqs. \ref{eq:Mi}, \ref{eq:Mlast} 
and the fact that
\begin{align*}
&1- \frac{F_{d,l}(y_{j})-F_{d,l}(y_{j}-1)}{1- F_{d,l}(y_{j}-1)}
=\frac{1-F_{d,l}(y_{j})}{1- F_{d,l}(y_{j}-1)},
\end{align*}
Eq. \ref{eq:M} can be simplified to}
\begin{align}
\label{eq:Mfinal}
&P\big(X_1 = (\delta_1,...,\delta_\gamma) | X_0= d,L_d=l\big) \nonumber \\
=& \left( \sum_{i=1}^{\gamma'-1}  {C_{i-1}\choose c_i} \left(F_{d,l}(y_i)-F_{d,l}(y_i-1)\right)^{c_i} \right) \\
&\cdot \left( \left(1- F_{d,l}(y_{\gamma'})\right)^{C_{\gamma'-1}} 
 - \sum_{j=0}^{c_{\gamma'}-1} {C_{\gamma'-1}\choose j} \right. \nonumber \\ 
&\quad \cdot  \left(F_{d,l}(y_{\gamma'})-F_{d,l}(y_{\gamma'}-1)\right)^{j}
\left(1- F_{d,l}(y_{\gamma'})\right)^{C_{\gamma'-1}-j}\Bigg) \nonumber
\end{align}

\schange{}{We can now determine the missing term $P(X_1=s|X_0=d,B=l)$ in Eq. \ref{eq:T1total},
thus completing the derivation of the closest contacts distribution.}

\subsection{Derivation of $I$}
\label{secI}
The derivation of the initial distribution $I$ 
requires the closest contact distribution as derived above and the distribution
of $X_0$.
For any state $s\in S$ with initial probability
$I(s)$, we have
\begin{align}
\label{eq:init}
I(s) = \sum_{d=0}^b P(X_1 = s | X_0=d) P(X_0=d).
\end{align}
The probability that a random requesting node $r$ has distance $d$
to $t$ corresponds to the fraction of IDs with this distance,
hence
\begin{align}
\label{eq:X0}
P(X_0=d) = \begin{cases}
\frac{1}{2^b}, \quad & d = 0 \\
\frac{2^{d-1}}{2^b}, \quad & d > 0 
\end{cases}.
\end{align}

\subsection{Derivation of $T$}
\label{sec:T}

The derivation of $T$ is more complex, but it is based on similar
concepts as earlier steps.
Let $A_0$ be the random variable for the current state, and $A_1$ the next state.
The transition probability
$P(A_1 = s | A_0=s_0)$
is derived for all $s_0, s \in S$.
The probability of the transition from $s_0$ to $s$ is given by 
first considering all possible sets of $\alpha \beta$ returned contacts
for state $s_0$.
For each set of returned contacts,  the probability distribution over
the set of distinct contacts needs to be derived.
\schange{Problems arise if less than $\alpha$ distinct new contacts are
returned. At this point the memoryless Markov chain model fails
to model the selection of a replacement contact, which has been returned in
earlier steps. 
Nevertheless, one can solve this rare case by using upper and lower
bounds on the distance of the critical contact to $t$.}{}

We start by considering case $A_0=\found$, for which 
\begin{align*}
P(A_1 = s | A_0=\found) = \begin{cases}
1, \quad s = \found \\
0, \quad s \neq \found
\end{cases}
\end{align*}
holds. The remaining entries of $T$ are of the form
$P\left(A_1 = s | 
A_0 = (d_1, \ldots ,d_\alpha)\right),$
where the next state $s$ is either $\found$ or a vector consisting of the distances of the $\alpha$ closest nodes queried in the next step.
The probability $P\left(Z^j=s^j | A_0(j)=d_j \right)$ for the state $Z^j=s_j=(s_j^1,\ldots , s_j^\beta)$ of the closest $\beta$ nodes in the routing table of the $j$-th queried node $v_j$
is given by Eq. \ref{eq:T1total}.
By Assumption $5$, routing tables are chosen independently, so that
\begin{align}  
\label{eq:independence}
\begin{split}
&P\left(Z^1 =  s^1,\ldots , Z^\alpha = s^\alpha | A_0  = (d_1, \ldots ,d_\alpha)\right) \\
&= \prod_{j=1}^\alpha P\left(Z^j=s^j | A_0(j)=d_j \right).
\end{split}
\end{align}
\schange{Again, we start by computing the probability of termination.}{} 
For each of the $\alpha$ considered contacts, the probability of termination is obtained from Eq. \ref{eq:F1},
using the bucket size $k_{d_j}$ and the $d_j$-th row of the matrix $L$.
The probability to terminate in the next step is then given as the complement of the event that
none of the parallel lookups terminates, i.e.
\begin{align}
\label{eq:F2}
\begin{split}
&P\left(A_1 = \found |  A_0 = (d_1, \ldots ,d_\alpha)\right) \\
&= 1 - \prod_{j=1}^\alpha \Big( 1 -P\left(Z_j = \found | A_0(j)=d_j \right) \Big).
\end{split}
\end{align}
If routing does not terminate,
it remains to obtain the closest $\alpha$ contacts from a set of $\alpha\beta$ returned
contacts. 
\schange{The derivation is complicated by the fact that a node can be returned multiple times, or be identical to an already
contacted node.}{}
\schange{The remaining part of these section mostly deals with the distribution of
distinct returned contacts.}{}
Let $\Gamma=(s^1_1,\ldots ,s^1_\beta, \ldots , s^\alpha_\beta)$ be the distances of the returned contacts.
\schange{In case a contact is returned multiple times, all but one of the $s^j_i$ associated with this contact are replaced by
a constant $K$.
If a contact is identical to one that has been considered, its distance is replaced by the distance
of an earlier returned candidate.
However, since the model does not provide information about such contacts, 
a distance $K^*$ is chosen, which provides an upper or lower bound. 
All known but not contacted nodes have distance at least $d_\alpha$,
so that for an upper bound on the success probability, we minimize the distance
of a replacement contact by $K^{up} = d_\alpha$.
In contrast, for a lower bound on the success probability, $K^{low}=b$ is chosen,
corresponding to a replacement node with maximal distance to $t$. }
{Due to the Markov property, we can only determine upper and lower bounds
on the distance of replacement contacts from earlier steps or the requester's routing
table. 
All known but not contacted nodes have distance at least $d_\alpha$,
so that for an upper bound on the success probability, we minimize the distance
of a replacement contact by $K^{up} = d_\alpha$.
In contrast, for a lower bound on the success probability, $K^{low}=b$ is chosen,
corresponding to a replacement node with maximal distance to $t$.
In the following, definitions and formulas specific to the upper bound are identified by the superscript
$up$, whereas the superscript $low$ characterizes the lower bound.
We use * to mean either $low$ or $up$.}

Denote by 
\begin{align*}
U^*(\Gamma) = \{ u = (u^1_1, \ldots , u^{\alpha}_\beta): u^j_i \in \{s^j_i,K^* \}\}    
\end{align*}
all possible
sets of distances of distinct contacts given the distances $(s^1_1,\ldots , s^\alpha_\beta)$.
So, in general, we obtain the transition probabilities as
\begin{align}
\label{eq:distinct}
&P\big(A_1  = (\delta_1,...,\delta_\alpha) | 
(A_0 = (d_1, ...,d_\alpha)\big) \nonumber \\
=& \sum_{\Gamma} \sum_{u \in U_\delta(\Gamma)} P^*(u | \Gamma) \\
& \prod_{j=1}^\alpha P\Big((Z^j(1),\ldots ,Z^j(\beta))=(s^j_1,...,s^j_\beta)| A_0(j)=d_j\Big) \nonumber
\end{align}
with 
$U_\delta(\Gamma) = \{u \in U^*(\Gamma): top_\alpha (u) = (\delta_1,...,\delta_\alpha) \}$.
In the following, we derive the probability $P^*(u | \Gamma)$
for each $u \in U^*(\Gamma)$. 
\schange{For a formal description, a considerable amount of notation needs 
to be introduced. }{}
The basis idea is to first find a maximal set $Y^*$
of definitive distinct contacts, and then iteratively determine for each
remaining  element the probability to be distinct from all elements in $Y^*$
as well as contacts queried earlier during routing.
The probability $P^*(u | \Gamma)$ for $u \in U(\Gamma)$ in Eq.  \ref{eq:distinct} is inductively computed, \schange{iterating over all $s^j_i \in \Gamma\setminus Y^*$}{conditioning
on $Y^*$}. 
More precisely, we transform
\begin{align}
\label{eq:iteration}
&P^*(u=(u^1_1,\ldots , u^\alpha_\beta) | \Gamma) \nonumber \\
&= P^{*}(u^1_1 | Y^{*},Z)P^{*}(u^1_2 | Y^{*},\Gamma,u^1_1)\\
&\cdots P^{*}(u^1_\beta | Y^{*},\Gamma,u^1_1,\ldots , u^1_{\beta-1}) P^{*}(u^2_1 | Y^{*},\Gamma,u^1_1, \ldots , u^1_\beta) \nonumber \\
&\cdots P^{*}(u^\alpha_\beta | Y^{*},\Gamma,u^1_1, \ldots , u^\alpha_{\beta-1}) \nonumber.
\end{align}
and determine each factor.
For a distance $a$ and a queried contact $v_j$, let $y^j_a$ be the set of indexes $(j,i)$, so that $s_j^i=a$.
\schange{Furthermore, $y^{max}_a = argmax \{ |y^j_a|: j=1\ldots \alpha\}$ contains a maximum number of distinct returned contacts with distance $a$.}
{Since each set $y^j_a$ contains from the routing table of one node, these are distinct, and
$y^{max}_a = argmax \{ |y^j_a|: j=1\ldots \alpha\}$ contains the maximal number of contacts with
distance $a$ that are guaranteed to be unique.}
So all contacts in
$Y^{low}= \bigcup_{a=0}^{d_\alpha-1} y^{max}_a$
\schange{are necessarily distinct.}{are unique and have not been contacted before because $d_1$ is
the minimal distance of all nodes contacted up to this point.}
In contrast, for the upper bound, earlier steps are not considered for computing the probability of a contact to be distinct,
i.e.
$Y^{up}= \bigcup_{a=0}^{b} y^{max}_a$.

The probability that \schange{a node at distance $s^j_i$}
{the $i$-th node returned by $v_j$ and having distance $s^j_i$ to $t$} is identical to an earlier considered one
is given as the ratio of contacted nodes at distance $s^j_i$ and all nodes at distance
$s^j_i$.
Consequently, we first compute the number of nodes $count(s^j_i, Y^{*},\Gamma,u^1_1, \ldots , u^j_{i-1})$ at distance $s^j_i$
that have been contacted and may be identical.
\schange{ The remaining number
of nodes with distance $s^j_i$ is then again a binomially distributed random variable.}{}
\schange{If $s^j_i < K^*$, we only }{If only nodes from the current set of returned
contacts are considered, i.e. for the upper bound or if $s^j_i < d_1$, let} 
\begin{align}
\label{eq:count1}
\begin{split}
&count\left(s^j_i, Y^{*},\Gamma,u^1_1, \ldots , u^j_{i-1}\right) 
= |y^{max}_{s^j_i}| 
+ |\{ (\gamma,\mu):  \\
& \quad s^j_i = s^\gamma_\mu = u^\gamma_\mu,\gamma < j \}| 
- |\{ \mu: s^j_i = s^j_\mu, u^j_\mu = r, \mu < i \}|
\end{split}
\end{align}
be the number of returned contacts that are potentially identical to the $i$-th
returned contact of $j$-th queried node $v_j$.
\schange{These consist of the maximum definitively distinct contacts and further distinct 
contacts in $u$. }{These consists of all returned contacts that have been decided to be unique
up to this point.}
The subtraction follows from the fact that $v_j$'s returned contacts are distinct.
So if a different contact returned by $v_j$ is identical to some contact $w$, we know that the $i$-th
contact is not identical to $w$.
On the other hand, \schange{if $s^j_{i} \geq K$}{if we are considering the lower bound and $s^j_i \geq d_1$}{} , we set 
\begin{align}
\label{eq:count2}
\begin{split}
&count(s^j_i, Y^{*},\Gamma,u^1_1, \ldots , u^j_{i-1}) = \alpha b,
\end{split}
\end{align} 
since each parallel lookup is guaranteed to terminate after maximally $b$ steps.
The non-contacted number of nodes $X^j_i$ at distance $s^i_j$ is
$B(n-\alpha\beta, 2^{s^j_i-1-b})$ distributed.
Using the above terminology,
\begin{align}
\label{eq:pu1}
&P^{*}(u^j_i=s^j_i | Y^{*},\Gamma,u^1_1, \ldots , u^{j-1}_{\beta-1},u^j_1,\ldots u^j_{i-1}) \nonumber \\
&= \sum_{m=0}^{n-\alpha\beta} P(X^j_i=m)\frac{m}{m+count(s^j_i)} \\
&=  \sum_{m=0}^{n-\alpha\beta} {n-\alpha\beta \choose m}
\left(2^{s^j_i-1-b}\right)^m \left(1-2^{s^j_i-1-b}\right)^{n-\alpha\beta-m} \nonumber \\
& \quad \frac{m}{m+count(s^j_i, Y^{*},\Gamma,u^1_1, \ldots , u^j_{i-1})} \nonumber
\end{align}
for $(j,i) \notin Y^{*}$ and by construction
\begin{align}
\label{eq:pu2}
P^{*}(u^j_i=s^j_i | Y^{*},\Gamma,u^1_1, \ldots , u^{j-1}_{\beta-1},u^j_1,\ldots u^j_{i-1}) =1
\end{align}
if $(j,i) \in Y^{*}$.
Inserting Eq. \ref{eq:pu1} and Eq. \ref{eq:pu2} in Eq. \ref{eq:iteration}, the remaining
term $P(u|\Gamma)$ in Eq. \ref{eq:distinct} is determined. 
Eq. \ref{eq:distinct} now gives the transition probabilities for general queries with the
goal of finding a lower or upper bound on the hop count distribution.
This completes our derivation of $T^{low}$ and $T^{up}$. 

\subsection{Summary}
We have modeled the hop count distribution in a Kademlia-type system as a Markov
chain with an $\alpha$-dimensional state space corresponding to the
$\alpha$ contacted nodes in each step.
We derived an initial distribution $I$ on the closest contacts in the requester's
routing table and transition matrices $T^{low}$ and $T^{up}$
for upper and lower bounds on the hop count distribution.
The fraction of queries that need at most $i$ steps is consequently bounded
by computing the distributions $\Hop^{up}_i$ and $\Hop^{low}_i$
and choosing the entry corresponding to $\found$.

Note that there are various possibilities to map the transition probabilities in 
Eq. \ref{eq:distinct} to entries in the matrices.
Any  bijective mapping from $S$ to $\mathbb{Z}_{|S|}$ (i.e. the row/column index)
is suitable.
On the basis of such mapping, we analyze the storage and computation complexity in the next section. 

\section{Model Complexity}
\label{sec:analysis}

In the first part of this section, we determine the space  
and computation complexity of deriving the hop count distribution.
Finding that the complexity is at least $\calO(b^\alpha)$, an evaluation of
the accuracy of smaller ID spaces than the common 128 or 160 bits is considered.

\subsection{Space complexity}
We assume that the whole matrix $T$ needs to be stored, without any memory
enhancements.

\begin{lemma}
\label{thm:storage}
The storage complexity for computing the hop count distribution of a 
\kad{b}{\alpha}{\beta}{k}{L}-system is $\mathcal{O}\left(\frac{1}{(\alpha!)^2}b^{2\alpha}\right)$.
\end{lemma}
\begin{proof}
The storage complexity is dominated by the matrix $T \in \mathbb{R}^{|S|^2}$.
Consequently, $|S|$ needs to be determined.
\begin{align*}
\begin{split}
 |S|  
&=|\{\found\} \cup \{ s \in \mathbb{Z}^\alpha_{b+1}: s_j \leq s_{j+1}, j=1\ldots \alpha-1 \}|\\
&=1 + \sum_{i_\alpha=0}^{b} \sum_{i_{\alpha-1}=0}^{i_\alpha} \cdots \sum_{i_1=0}^{i_2} 1 \\
&= \mathcal{O}\left( \int_0^b \int_{0}^{x_\alpha} \cdots \int_{0}^{x_2} 1 dx_1dx_2\ldots dx_\alpha\right)\\
&= \mathcal{O}\left( \frac{1}{\alpha!} b^\alpha \right)
\end{split}
\end{align*}
The size of the matrix $T$ is $S^2$ and by this
the space complexity is $\mathcal{O}\left(\frac{1}{(\alpha!)^2}b^{2\alpha}\right)$ as claimed.
\end{proof}

\subsection{Computation complexity}
We bound both the computation complexity in terms of necessary basic operations of the lower as well as the upper bound.

\begin{lemma}
\label{thm:complex}
The computation complexity is linear with regard to the network order $n$, and
polynomial with regard to the bit number $b$.
 More precisely the number of basic operations is 
of order \mbox{$\calO\left(n b^{\alpha (\beta + 2)} \right)$}.
\end{lemma}
\begin{proof}
We need to analyze the computation costs for the initial distribution $I$, the transition matrix
$T$, and the matrix multiplication for obtaining $P_i$.

Note that in case of both $I$ and $T$, the computation of $\gamma \in \{\alpha, \beta \}$ closest neighbor distribution
is essential.
The success probability given in Eq. \ref{eq:F1} can be determined in $\calO\left( n \right)$
if binomial coefficients are computed iteratively. Note that this has to be done for $d=1,\ldots , b$,
resulting in a cost of 
\begin{align}
\label{eq:foundcompute}
H_{\found} = \calO(n \cdot b).
\end{align}
These probabilities can be precomputed and stored, as can the values of
the cumulative distributions $F_{d,l}$ and the binomial coefficients used in Eq. \ref{eq:Mfinal}.
Using iterative computations of powers and binomial coefficients, the cost for these computations is
\begin{align}
\label{eq:cdfcompute}
H_{P} = \calO(b^3 + \max\{\alpha, \beta \}^2) 
\end{align}
The term $b^3$ for the CDF computations follows because there are $\calO(b^2)$ functions (one for
each $d,l \in \mathbb{Z}_{b+1}$), each taking up to $b$ distinct values.
Assuming precomputation, one evaluation of Eq. \ref{eq:Mfinal} has computation cost
$\calO(\gamma \kappa)$ for  $\kappa = \max\{ k_d: d=0 \ldots b \}$.
To see this, consider that the sum from $1$ to $\gamma'-1$ has at most $\gamma-1$ summands
that are products of terms $F_{d,l}(y_i)-F_{d,l}(y_i-1)$ and binomial coefficients.
The remaining factor for the last term has at most $\gamma$ summands, each consisting of at
most $k_d+1$ factors. 
Note that for each pair $d$, $l$, there are $d-l$ distances
a node
in the respective bucket can have.
Therefore, the number of evaluations for a given distance $d$
 is, similarly to Lemma \ref{thm:storage}, bounded by
\begin{align}
\label{eq:nrofc}
\begin{split}
& \sum_{l=1}^d \sum_{\delta_\alpha=0}^{d-l} \sum_{\delta_{\alpha-1}=0}^{\delta_\alpha} \cdots \sum_{\delta_1=0}^{i_2} 1 \\
=& \calO\left( \int_0^d \int_0^{y-z} \int_{0}^{x_\alpha} \cdots \int_{0}^{x_2} 1 dx_1dx_2\ldots dx_\alpha dz \right) \\
=& \calO\left( \int_0^d\frac{1}{\gamma !} (d-z)^\gamma dz \right) 
= \calO\left(  \frac{1}{(\gamma+1) !} d^{\gamma+1} \right)
\end{split}
\end{align} 

For the initial distribution, the computation cost is hence given by
\begin{align}
\label{eq:initcompute}
H_I = \calO\left(  \frac{1}{(\alpha+2) !} b^{\alpha+2} \alpha \kappa\right),
\end{align}
summarizing over $d=1\ldots b$.
The additional computation cost of the requester's distance distribution $X_0$ in Eq. \ref{eq:X0}
is $\calO(b)$ (as usual, we assume that powers of two are calculated iterative), which is clearly dominated by the
computation cost of closest neighbor distribution. 

In contrast, for computing the transition matrix, one first has to consider all possible $\alpha\beta$ returned values for each state $A_0=(d_1,\ldots , d_\alpha)$. 
The number of returned sets is determined based on  Eq. \ref{eq:nrofc} with $\gamma=\beta$.
\begin{align*}
\begin{split}
&\calO\left( \sum_{d_\alpha=0}^{b} \frac{1}{(\beta+1)!}d_\alpha^{\beta+1}  \cdots \sum_{d_1=0}^{d_2} \frac{1}{(\beta+1)!}i_1^{\beta+1}\right) \\
=&\calO\left(\frac{1}{((\beta+1)!)^\alpha} \int_{0}^{b} x_\alpha^{\beta+1}  \cdots \int_{0}^{x_2} x_1^{\beta+1} dx_1 \ldots dx_{\alpha}\right) \\
=& \calO\left(\frac{1}{((\beta+1)!)^\alpha} \int_{0}^{b} x_\alpha^{\beta+1}  \cdots \int_{0}^{x_3} \frac{1}{\beta+2}x_2^{2\beta+3} dx_2 .. dx_{\alpha}\right) \\
=& \calO\left(\frac{1}{((\beta+1)!)^\alpha} b^{\alpha \beta + 2\cdot \alpha} \prod_{j=1}^\alpha \frac{1}{j\beta + 2\cdot j}\right) \\
=& \calO\left(\frac{1}{((\beta+1)!)^\alpha} b^{\alpha (\beta + 2)} \prod_{j=1}^\alpha \frac{1}{(\beta + 2)j}\right)
\end{split}
\end{align*} 
For each set of returned contact all possible distinct sets have to be evaluated.
This results in a factor $\calO\left( 2^{\alpha\beta} \right)$. 
The probability of each contact being distinct is determined in $\calO(n)$ operations by Eq. \ref{eq:pu1},
however, it is possible to precompute the probabilities for all distances $d$ and $(\alpha-1)\beta+1$ values of 
count (Eq. \ref{eq:count1} and Eq. \ref{eq:count2}).
So an additional cost of $\calO\left(n\alpha\beta 2^{\alpha\beta} \right)$
per set of returned contacts is needed.
Furthermore, for each of these combinations the function $\top_\alpha$ is applied, which is a factor of $\calO \left(\alpha^2 \beta\right)$. The total complexity of transition matrix computation is hence
\begin{align}
\label{eq:tcompute}
\begin{split}
&H_{T} = \calO\left( \frac{1}{((\beta+1)!)^\alpha} b^{\alpha (\beta + 2)} \prod_{j=1}^\alpha \frac{1}{(\beta + 2)j} 
\beta \kappa n\alpha\beta 2^{\alpha\beta}  \right) . 
\end{split}
\end{align}
It remains to determine the overhead of matrix multiplication.  
The target is definitively reached after at most $b+1$ steps, since the distance of the first lookup decreases
by at least one in each step. 
Each matrix multiplication takes $|S|^2$ operations.
By the proof of Lemma \ref{thm:storage}, the complexity of the matrix operations is hence
\begin{align}
\label{eq:multcompute}
H_{M} = \calO\left( |S|^{2}b\right) = \calO\left(\frac{1}{(\alpha!)^2} b^{2\alpha+1}\right)
\end{align}

Summarizing over all computation costs, the total complexity is
\begin{align*}
\begin{split}
& H_{\found} + H_{P} +H_I + H_T + H_M \\
&=\calO\left( nb + b^3 + \max\{\alpha, \beta \}^2  + \alpha \kappa \frac{1}{(\alpha+2) !} b^{\alpha+2}  \right. \\
& \left. \quad + \frac{1}{((\beta+1)!)^\alpha} \prod_{j=1}^\alpha \frac{1}{(\beta + 2)j} 
 \kappa \alpha\beta^2 2^{\alpha\beta}  n b^{\alpha (\beta + 2)} \right. \\
 & \left. \quad + \frac{1}{(\alpha!)^2} b^{2\alpha+1} \right)
\end{split}
\end{align*}
basic operations by Eqs. \ref{eq:foundcompute}, \ref{eq:cdfcompute}, \ref{eq:initcompute}, \ref{eq:tcompute}, and
\ref{eq:multcompute}. 
Treating $\kappa$, $\alpha$, and $\beta$ as constant factors gives the claimed complexity.

\end{proof}

\subsection{Reducing the ID space size}
\label{sec:accuracy}

From Lemma \ref{thm:storage} and \ref{thm:complex}, we can see that both the storage and the computation complexity is polynomial in the bit-size $b$ for a relative high degree polynomial. 
In contrast, the dependence on the network order $n$ is only linear for the computation complexity, whereas the storage complexity
is independent of $n$. 
Though the dependence on $\alpha$ and $\beta$ is exponential, both can be assumed to be small.
For instance, if $\alpha=3$, the number of entries in the matrix $T$ can be precisely computed as
\begin{align*}
\begin{split}
&|\{F\} \cup \{(s_1,s_2,s_3) \in \mathbb{Z}_{b+1}^3: s_1 \leq s_2 \leq s_3  \}|^2 \\
&=\left(1 + \sum_{i_3=0}^b \sum_{i_2=0}^{i_3} \sum_{i_1=0}^{i_2} 1\right)^2 \\
&=\left( 1 + \sum_{i_3=1}^{b+1} \frac{i_3(i_3+1)}{2}\right)^2\\
&=\left( 1 + 0.5 \cdot \sum_{i_3=1}^{b+1} (i_3^2 + i_3) \right)^2\\
&=\left( 1 + 0.5 \cdot \left(\frac{(b+1)(b+2)(2b+3)}{6} + \frac{(b+1)(b+2)}{2}\right) \right)^2\\
&=\left( 1 + \frac{(b+1)(b+2)(2b+6)}{12}  \right)^2. 
\end{split}
\end{align*}
In practice, $b=128$ and $b=160$ are typically used, corresponding to the length of MD-5 and
SHA hashes.
For these sizes, the matrix $T$ has $134062161025$ and $502058690721$ entries, respectively.
Assuming 32 bit float numbers, this amounts to roughly $500GB$ and  $1870GB$, respectively.

Consequently, the computations are too expensive to present an alternative to
extensive simulations.
However, the dependence on the actual ID space size can be expected to be small, at least if
the number of IDs decisively higher than the number of nodes.
The following Lemma provides an upper bound on the influence of $b$.

\begin{lemma}
Consider two Kademlia-type systems $K=$\kad{b}{\alpha}{\beta}{k}{L} and 
$\tilde{K}=$\kad{\tilde{b}}{\alpha}{\beta}{\tilde{k}}{\tilde{L}}, such that
\begin{itemize}
\item $\tilde{b} < b$
\item $k[0..\tilde{b}] = \tilde{k}$, i.e. the vector $\tilde{k}$ contains exactly
the  $\tilde{b}+1$ first entries of $k$
\item $L[0..\tilde{b}][0..\tilde{b}] = \tilde{L}$, i.e. the matrix $\tilde{L}$ is the
square matrix containing the first $\tilde{b}+1$ first entries of $L$.
\end{itemize}
The fraction of terminated queries after $i$ hops in $K$ and $\tilde{K}$ differs by at most
\begin{align}
\label{eq:accuracy}
|P^*(i) - \tilde{P}^*(i)| \leq 1-\sum_{j=0}^{\kappa} {n \choose j} p^j \left(1-p\right)^{n-j}
\end{align}
with $p=2^{-\tilde{b}}-2^{-b}$, $\kappa = \min \{ k_d:d=0\ldots b \}$, and $*\in \{up,low\}$.
\end{lemma}
\begin{proof}
Note that the two systems only behave different in case there at least $\kappa$ nodes that share a common prefix of length at least 
$\tilde{b}$ with the target $t$, but not with a common prefix length of $b$. 
Recall that the query is considered successful in the next step after reaching a node with common prefix length $\tilde{b}$ by Assumption
6 (see Section \ref{sec:assumptions}) in $\tilde{K}$, but not necessary in $K$.
The probability that two nodes share a common prefix of length $\tilde{b}$ to $b-1$ is 
$p=2^{-\tilde{b}}-2^{-b}$. The number of nodes with this
property is hence $B(n,p)$ distributed.
The claim follows directly.
\end{proof}

Based on Eq. \ref{eq:accuracy}, we can consider the trade-off between accuracy and computation speed in terms of the network order $n$.

\begin{theorem}
\label{thm:ncomplexity}
When the error is supposed to be bounded by $\delta$ for some $C>0$, 
$\tilde{b}=\lceil  \log_2 ( n \frac{2}{\ln 1/\delta} ) \rceil$ 
achieves the required accuracy.

Consequently, the storage complexity is $\calO\left( log^{2\alpha} n \right)$ and the computations complexity
is $\calO\left( n polylog (n) \right)$ for a constant arbitrary small error $\delta$. 
\end{theorem}
\begin{proof}
For $\kappa=0$ in Eq. \ref{eq:accuracy}, we can determine an upper bound on the minimal value for $\tilde{b}$
to achieve an error of less than $\delta$. Set $C=\frac{2}{\ln 1/\delta} $ in the following.
From
\begin{align*}
\delta \leq 1 - (1-p)^n < 1 - (1-2^{-\tilde{b}})^n,
\end{align*}  
it follows that
\begin{align*}
\tilde{b} \geq \log_2 \frac{1}{1- \delta^{1/n}}.
\end{align*}
It remains to show that for $n$ large enough,
\begin{align*}
\frac{1}{1- \delta^{1/n}} < Cn 
\end{align*}
Rewriting results in $\delta < (1 - \frac{1}{Cn})^n$. 
Because $(1 - \frac{1}{Cn})^n$ converges to $e^{-1/C}$, there exists $n$, such that
\begin{align*}
\left(1 - \frac{1}{Cn}\right)^n > e^{-2/C} = e^{-ln 1/\delta} = \delta
\end{align*}
Therefore, for $n$ large enough $\tilde{b} \geq \log_2 \left( n \frac{2}{\ln 1/\delta} \right)$
ensures that $|P^*(i) - \tilde{P}^*(i)| \leq \delta$.

The storage complexity is $\calO\left( log^{2\alpha} n \right)$ by Lemma \ref{thm:storage} with $\tilde{b}  = \calO\left( \log n\right)$,
the computation complexity of $\calO\left( n polylog(n) \right)$ follows from Lemma \ref{thm:complex}.  
\end{proof}

\vspace{7pt}
Note that by using Eq. \ref{eq:accuracy} rather than the approximation in Theorem \ref{thm:ncomplexity}, the bound on $\tilde{b}$
can be further reduced. However, Theorem \ref{thm:ncomplexity} proves that the number of bits needed for a certain accuracy grows at most logarithmically
in the network size. 

In this section, we have seen that the storage complexity of our analytical solution is polylog in the number of nodes, whereas simulation require at least a linear overhead. 
The computation complexity is slightly higher than linear, but simulations are bound to require a similar cost for establishing the routing tables, not even considering the actual routing performance.
Our results in Section \ref{sec:verification} and Section \ref{sec:results} show that our algorithm can easily compute hop count distributions for large network sizes.
\section{Verification and Scalability}
\label{sec:verification}

 \schange{In this section, we show that our model agrees with no-churn simulations.
 Furthermore, we show that it scales to networks of up to $10M$ nodes
 easily and that the differences between the upper and lower bounds are
 negligible.
 We conclude with a comparison to real-world measurements.}
 {In this section, we compare the model with static simulations as well as real-world
 measurements.}
 
 \begin{figure*}[ht]
\centering
\subfloat[$n=100K,\alpha=3,\beta=2$]{\label{fig:a3b2100k}\includegraphics[width=0.5\linewidth]{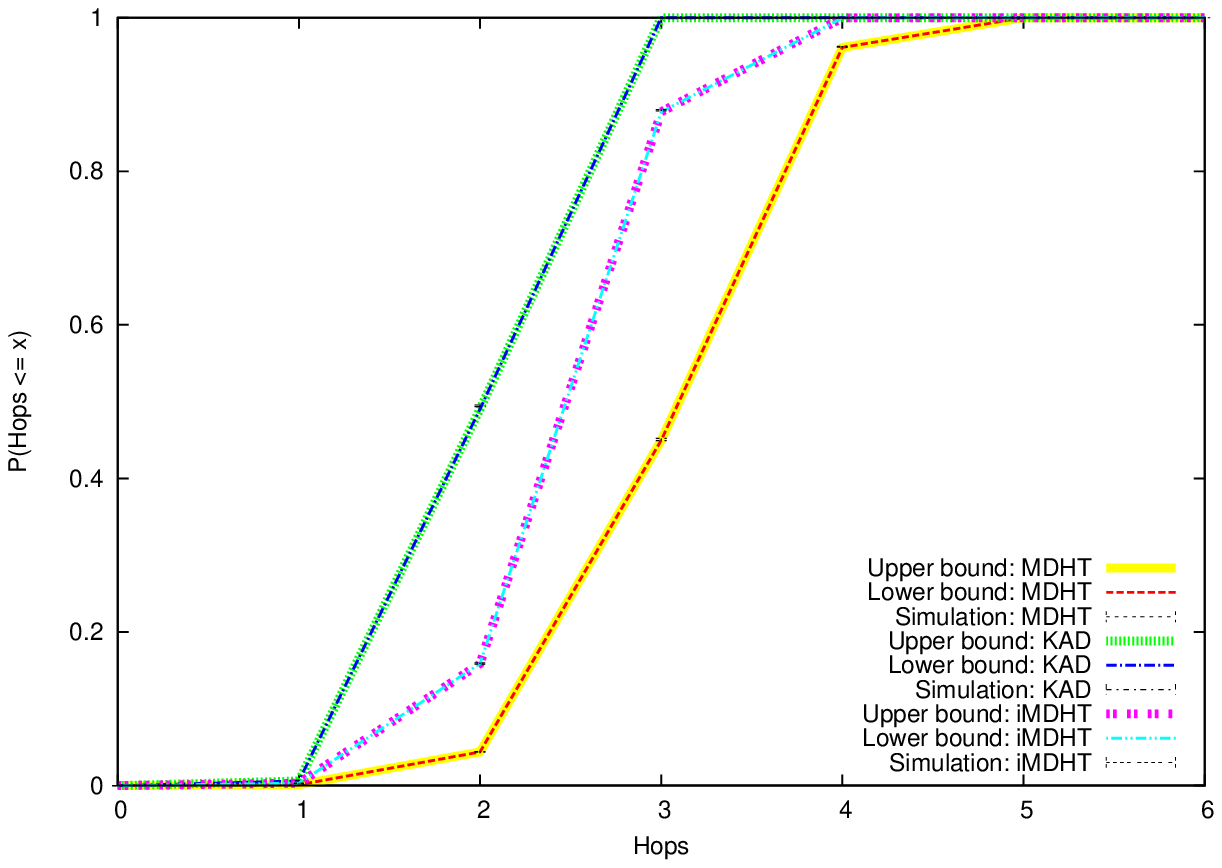}}
\subfloat[$n=100K, \alpha=4,\beta=1$]{\label{fig:a4b1100k}\includegraphics[width=0.5\linewidth]{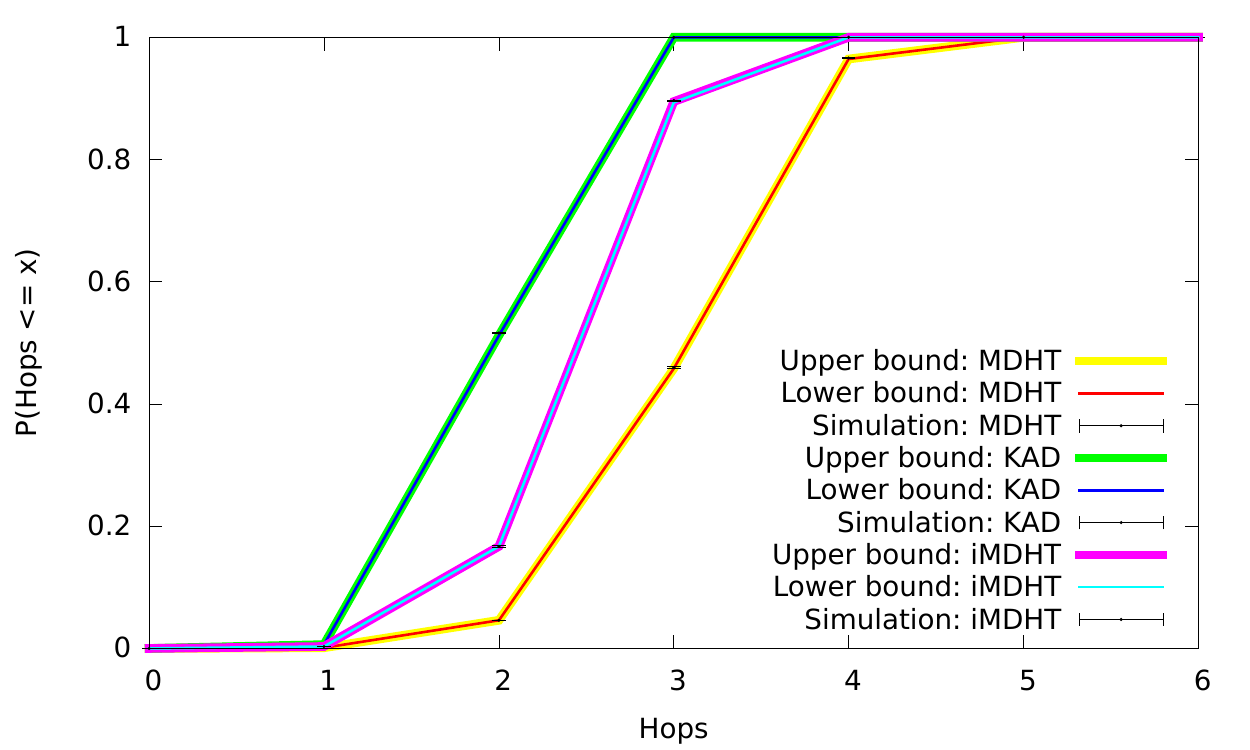}}
\caption{Model vs. Simulation: Cumulative hop count distribution of MDHT, iMDHT, and KAD ($100K$ nodes)}
\label{fig:cdfs}
\vspace{-2em}
\end{figure*}

\subsection{Model verification}
\label{sec:veri}
We have chosen three routing table structures: MDHT, iMDHT, and KAD, as they are described in Section \ref{sec:related}.
As for the routing parameters, we focus on the two settings which are used by widely-used Kademlia implementations: ($\alpha=3$, $\beta=2$) and ($\alpha=4$, $\beta=1$).

The error rate is chosen as $\delta=0.001$, which can be expected to
be far below the confidence intervals length of simulation results.  
By Eq. \ref{eq:accuracy}, the number of bits needed for the desired accuracy are 14 ($100K$, KAD),
15 ($100K$, MDHT and iMDHT), and 21 ($10M$). 
The first value is lower for KAD than MDHT and iMDHT, because KAD has a bucket size of $10$ rather than $8$. 

For validation, we use the simulation framework GTNA \cite{Schiller10gtna}, a tool for network analysis. GTNA offers a variety of metrics for analyzing graphs, as well as an easily extensible routing algorithm interface. 
We extended the framework by adding the MDHT, iMDHT, and KAD networks, as well as the Kademlia routing algorithm\footnote{The code is available at: \url{https://github.com/stef-roos/GTNA/tree/grouting}}. We choose GTNA rather than an event-based simulator for various reasons.
Most importantly, our initial model does not consider churn, failure, and varying latencies between peers. 
Including such behavior in the simulation environment will inevitably lead to derivations of
the analytic results and the observed hop counts. 
However, it is not possible to easily distinguish between these derivations and actual faults
in the model.
Using a network simulator without enabling real network conditions is an overhead with regard
to storage space, computation time as well as implementation complexity.
GTNA can easily scale to $1M$ nodes, which is hard to achieve by an event-based
simulator, e.g.  OverSim \cite{baumgart09oversim}.

The networks are generated as follows: Each node $v$ is given a 128-bit identifier, its routing table is constructed by first randomizing the list of nodes. 
Nodes in the list are considered iteratively and added to $v$'s routing table
if there is an empty slot in the corresponding bucket.
In this way, maximally full buckets are realized, which cannot be guaranteed by the 
real-world protocol.
The routing algorithm progresses step-wise, always querying $\alpha$ nodes and
processing all answers, before contacting the next set of nodes.
We generated $20$ network topologies uniformly at random, and routed to five distinct,
randomly selected, target nodes from each node's routing table.


Figures \ref{fig:a3b2100k} and \ref{fig:a4b1100k} show the resulting cumulative hop count distributions for a $100K$-node
network. 
Both the upper and lower analytic bounds are shown for the three networks and the two sets of
routing parameters. 
Furthermore, the $95\%$ confidence interval of the simulations is given.
Upper and lower bounds are extremely close (at most an absolute difference of $0.2\%$), and both are within the confidence interval of the simulation.
The negligible difference between the upper and lower bounds can be expected, seeing that the 
success probability for the first two steps is identical and most routes terminate within three hops.
The parameters $\alpha=3$,$\beta=2$ and $\alpha=4$,$\beta=1$ are hard to distinguish from the two graphics,
but the later achieves a slightly higher success rate for each hop.
We discuss the impact of the routing algorithms as well as the routing table structure in Section \ref{sec:results}.
All in all, the results show a strong agreement between the model and the simulations, indicating that the
derivation and implementation are indeed correct.

\subsection{Scaling to large networks}
\label{sec:scaling}

In this section, the hop count distribution is analytically computed for networks of $10M$ nodes.
Besides showing that the model easily scales, it is essential to see that the divergence between
upper and lower bounds is acceptable for giving meaningful performance bounds.
 \begin{figure*}[ht]
\centering
\subfloat[$n=10M,\alpha=3,\beta=2$]{\label{fig:a3b210m}\includegraphics[width=0.5\linewidth]{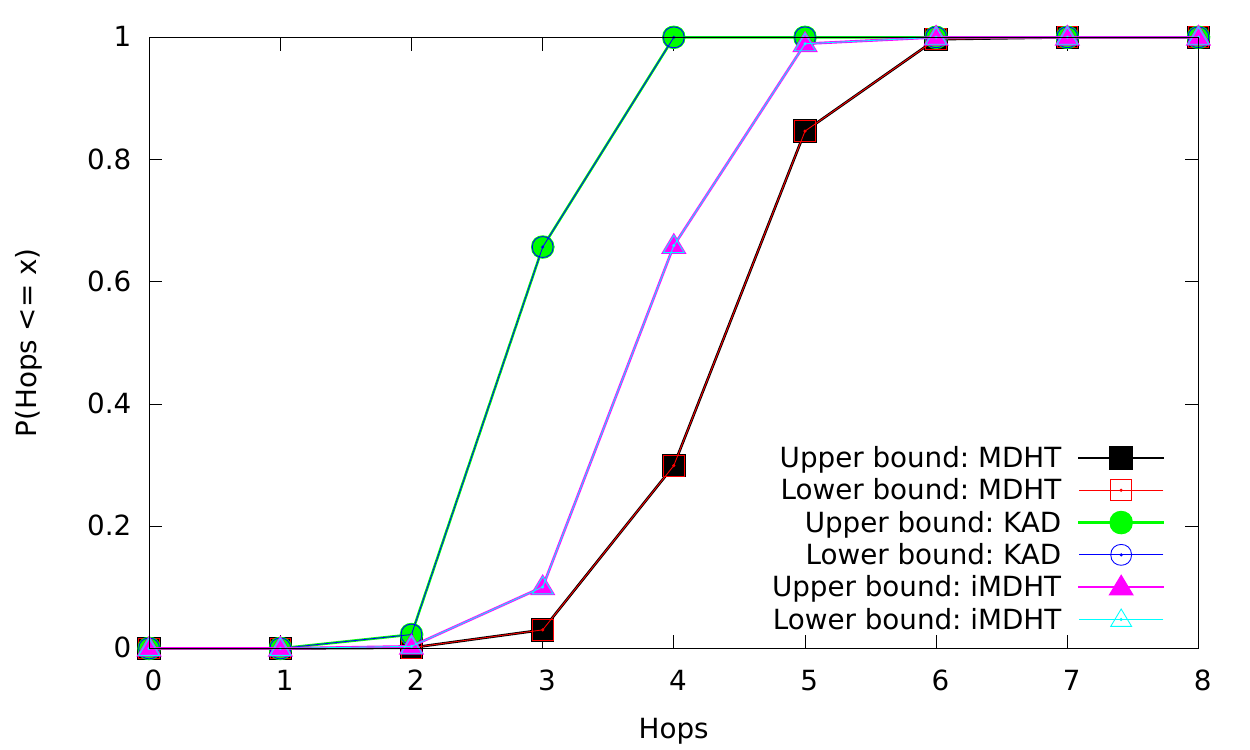}}
\subfloat[$n=10M, \alpha=4,\beta=1$]{\label{fig:a4b110m}\includegraphics[width=0.5\linewidth]{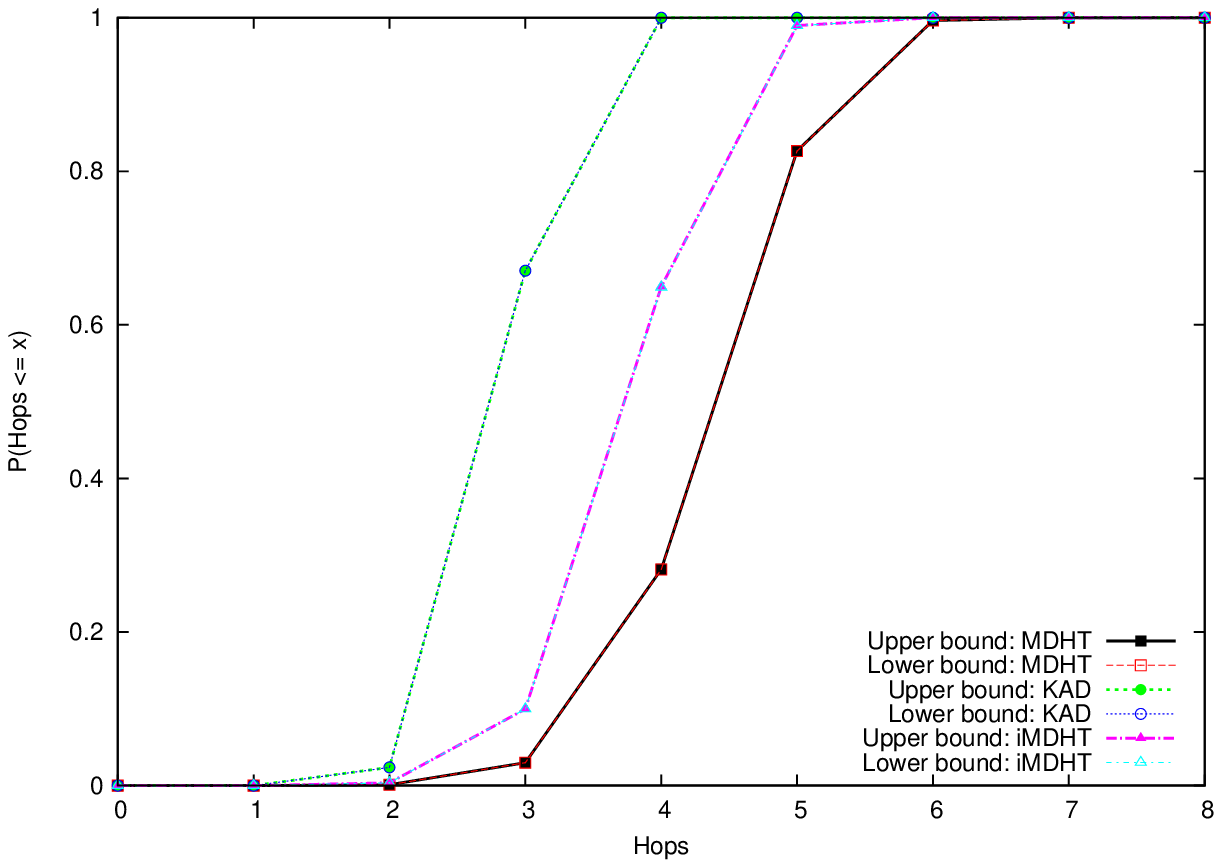}}
\caption{Cumulative hop count distribution of MDHT, iMDHT, and KAD ($10M$ nodes), Analytic bounds and simulation are extremely close and appear identical}
\label{fig:cdfs10M}
\vspace{-2em}
\end{figure*}
As can be seen in Figures \ref{fig:a3b210m} and \ref{fig:a4b110m}, upper and lower bounds remain extremely close, differing
at most by $0.5\%$ in case of MDHT using $\alpha=4$
and $\beta=1$. 

In general, we see that lookups terminate fastest in the KAD system, and faster in iMDHT than MDHT, for all considered network sizes and routing parameters.
Note that the difference between the two routing algorithms is more noticeable than for 100K nodes.
Interestingly, $\alpha=3$, $\beta=2$ achieves a higher success rate for MDHT from the third hop onwards.

\subsection{Real-world measurements}
\label{sec:realworld}

\schange{Having seen that our model provides reliable bounds for static networks, we had a look at the results of real-world measurements and
compared them to the results of our model.}
{After validation the model in a controlled environment, we compare our results with real-world measurements to see if
our assumptions are too far from reality to produce meaningful results.}
Because real-world KAD routing tables have been shown to contain a lot of stale entries as well as missing entries \cite{Stutzbach06improving, Steiner10eval}, it is expected
that the results of our model differ from the measured ones. 
 It remains to ascertain how large this divergence is.
In practice, hop counts have been measured in KAD of $1M$ nodes \cite{Stutzbach06improving}. The measured
average hop count of $3.08$ is reasonably higher than our prediction of $2.81$.
Note that we compare the results of our model to the result of \cite{Stutzbach06improving} rather than \cite{Steiner08faster, Steiner10eval} or \cite{Jimenez2011subsecond},
because the first considers locating files with a high number of replicates rather than
source-destination lookups and the latter only gives latencies.
The network has been found to have
about $10\%$ of stale contacts as well as about $15\%$ of missing entries, which are bound to slow down the routing process.   
As a consequence, obtaining reasonable bounds on deployed networks requires enhancing our model to deal with churn and routing tables incompleteness.
In the next section, we give an initial approach for dealing with failures and routing table incompleteness, which closes the gap between the model and reality.

   
\section{Extending the Model}
\label{sec:churn}
In this section, we exemplary show how to modify the model \schange{by adding
the probability that an entry is stale and thus does not return any contacts}
{to account for stale entries and bucket incompleteness}.

The model treats non-responding contacts as follows: With a probability of $1-p$, a queried node is online and returns new contacts following the distribution described in Section \ref{sec:model}, otherwise
there are no returned values for this contact. 
If less than $\alpha$ distinct contacts are returned altogether, the remaining distance values are chosen as the highest distance $d_\alpha$ of the currently queried nodes
(for an upper bound on the hop count) and the overall maximal distance of $b$ bits (lower bound), analogously to Section \ref{sec:model}.
Formally, the state $Off$ is added to characterize an unresponsive node. Assume that $X_j$ is the distance of the $j$-th closest contact $v_j$,
and $L_d$ is the guaranteed bit gain as in Section \ref{sec:model}.
Then the probability distribution of contacts $Y_j$ returned by $v_j$ is  given by
\begin{align}
& \quad P(Y_j = s | X_j = d, L_d=l) =  \\  
& \begin{cases}
p,  &s=Off \\
(1-p) P(Y_j = s | X_j = d, L_d=l, s\neq Off),  &s \neq Off
\end{cases}
.\nonumber
\end{align}
$P(Y_j = s | X_j = d, B=l, s\neq Off)$ is determined in
Eqs. \ref{eq:d0}, \ref{eq:F1}, and \ref{eq:Mfinal}.
One more change has to be made with regard to the model in Section \ref{sec:model}. The lower bound on the success rate relies on the fact that the distance to the target decreases in every step. This cannot be guaranteed because
of fall-backs due to failures. For this reason, a  hops-to-live counter $htl$, is added, which is the maximal number of routing steps until the query is aborted. As a result, at most $count = htl\cdot \alpha$ nodes can be contacted during routing, rather than the bounds provided by Eq.  \ref{eq:count1} and Eq. \ref{eq:count2}.
Note that this change only influences the lower bound on the success probability. 
For the upper bound, Eq.  \ref{eq:count1} still applies because all earlier steps are disregarded.

\schange{}{Our modification regarding bucket incompleteness is simple: We reduce the bucket size to a distance-dependent factor
$c[d]$ of its actual value. 
An accurate model is to use level-dependent distributions on the actual number of nodes per bucket, however,
it is unlikely to obtain representative data while designing the systems, so that aggregates 
offer probably a similar accuracy and reduce the complexity both in terms of comprehensibility and actual computation cost.}

The remainder of this section deals with the comparison of the extended model to measurements.
\schange{For a stale entry rate $p=0.1$ in a KAD network with $1M$ nodes,
the average hop count is bounded by $2.90$ and $2.91$ (with $htl=7$, which achieves $99.7 \%$ of successfully terminated
queries).}
{Without considering bucket incompleteness, the average hop count is bounded by $2.90$ and $2.91$ (with $htl=7$, which achieves $99.7 \%$ of successfully terminated
queries) for a stale entry rate $p=0.1$ in a KAD network with $1M$ nodes.}
This is still considerably lower than the bound of $3.08$ observed in \cite{Stutzbach06improving}. 
However, their measurements reveal that there are in average about 1.5 missing entries per bucket.
\schange{An exact model needs a distribution of missing entries on a per-level basis,
however, for a simple estimation, we only consider the average.
Then, we can simply reduce the bucket size rather than adding further variables
in form of bucket completeness distributions for each level. }{}
\schange{	Indeed, the average hop count
is increased $3.0$ for both upper and lower bound 
if the first 10 levels have a bucket size of $9$, whereas all others have $8$
(in rough agreement with the per-level averages in \cite{Stutzbach06improving}).
}{Indeed, the average hop count
is increased to $3.0$ for both upper and lower bound 
if the bucket size is reduced to $c[d]=0.9$ for $d=b-9 \ldots b$ and
$c[d]=0.8$ for $d<b-9$ of its actual value
(in rough agreement with the per-level averages in \cite{Stutzbach06improving}).}
The remaining difference can be explained by using averages for the missing entries rather than the actual distribution.
\schange{, the assumption that the target is known if it is within the $k$ closest contacts cannot be guaranteed
to hold (see Eq. \ref{eq:F1}).}{} 
Despite the expected discrepancy between theory and measurements, our model more than halves the error rate
from $5.5\%$, provided by Stutzbach and Rejaie when integrating all of the above in their analytic model,
to merely $2.67\%$.
In conclusion, our model can be extended to include failure rates and bucket incompleteness,
as long as rough estimates of these parameters are known.
\section{Lessons Learned}
\label{sec:results}

\begin{figure*}[ht]
\centering
\subfloat[Routing Algorithm Parameters]{\label{fig:routing}\includegraphics[width=0.5\linewidth]{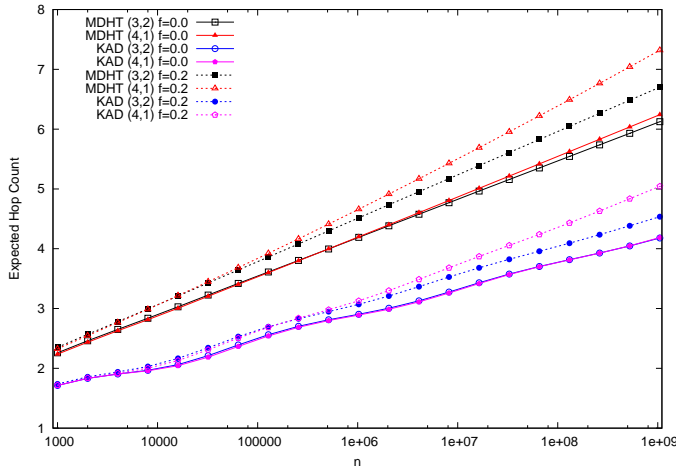}}
\subfloat[Multiple Buckets per Level]{\label{fig:resolution}\includegraphics[width=0.5\linewidth]{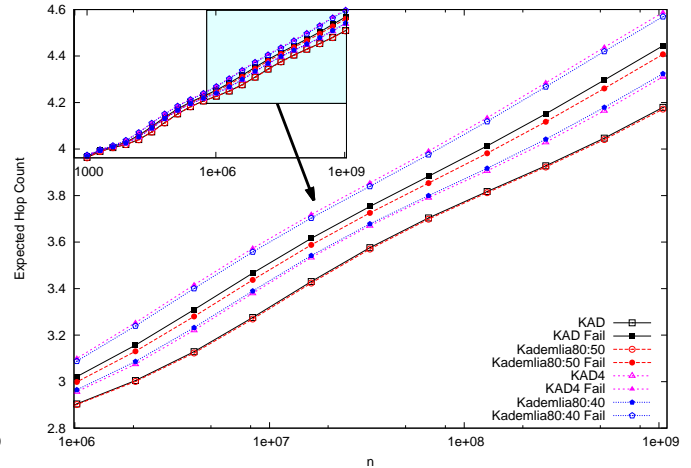}}
\caption{Comparison of routing parameters and routing table structures on the average hop count}
\label{fig:lessons}
\vspace{-2em}
\end{figure*}

In this section, we analyze the influence of the routing parameters $\alpha$ and $\beta$ as well as 
the routing table structure on the average hop count and the resilience to stale entries.
For this purpose, we denote the routing algorithm with parameters $\alpha$ and $\beta$ by $R_{\alpha,\beta}$.
Furthermore, MDHT and KAD refer to the routing table structures in the corresponding systems,
independent of the routing algorithm.
Networks of order $n=2^i \cdot 1000$ are evaluated for $i=0\ldots 20$, 
i.e. our model scales easily up to $1B$ nodes.
As in Section \ref{sec:verification}, the error rate is chosen as $\delta=0.001$.
Only upper bounds on the hop count are presented in favor of readability.
However, results for the lower bounds are very similar and entail the same conclusions.

For all considered parameters, the hop count has been shown to increase at most logarithmic
with $n=2^i$ \cite{Maymounkov02Kademlia}.
However,  if the number of contacts per level is limited by
$k$, the expected out-degree increases by $k$ whenever the number of nodes doubles, 
so that the average shortest
path length is of order $d=\calO\left(\frac{i}{\log i + \log k}\right)$
(solving $(ik)^d = 2^i$)
 rather than $i = \log n$.
Due to the extremely short routes observed in large networks, the hop count can be expected to 
follow a similar dependence. 
Indeed,  the sub-logarithmic routing complexity is visualized  as a slight curvature in the log plot (see Figure \ref{fig:lessons}), which is more noticeable
if the number of contacts per level is higher.

We start our evaluation of real-world systems by analyzing if 
 the change from $R_{3,2}$ to $R_{4,1}$ in MDHT actually decreases the
average hop count. 
In addition, we also consider the influence on the KAD routing
table structure.
In order to evaluate the impact of churn, the fraction $f$ of stale contacts is varied between $0.0$ and
$0.2$.
It can be expected that for smaller networks, the use of a higher value for $\alpha$ actually
increases the success rate because more routing tables are considered
in each hop.
However, when the network size increases, the high fall-back in case of duplicates or failures
for $\beta=1$ is bound to decrease the performance, 
so that there is a threshold from which on $R_{3,2}$ achieves shorter routes.
The advantage of $R_{3,2}$ is bound to be more obvious when the fraction of stale
entries is high because of the increased use of fall-back contacts.

Indeed, in the absence of churn, about half
a million nodes are needed for $R_{3,2}$ to have a lower hop count
for MDHT, 
but for KAD the threshold is only reached at half a billion nodes.
As expected, $R_{3,2}$ deals with churn more effectively due to the high number of
returned contacts to choose from.
Assuming a stale entry rate of $20 \%$, the average hop count in MDHT is increased
by up to $8 \%$ for $R_{3,2}$, but more than $12\%$ for $R_{4,1}$ assuming $8$ million participants
(about the estimated population of MDHT).
The relative performance degradation increases with the network size due to the
the longer routes, so that for $1B$ nodes, a divergence of up to $21\%$ 
occurs.


Interestingly, for common real-world networks of $1M$ to $10M$ nodes, the change from $R_{3,2}$ to
$R_{4,1}$ would have decreased the hop count in the KAD system for a very low stale entry rate, but not in MDHT, which actually
introduced this change. 
Note that the number of contacted nodes per hop is not increased by
a high value for $\beta$, but grows linearly with $\alpha$.
Overall, an increased number of returned contacts $\beta$ 
is preferable to an increased degree of parallelism $\alpha$
in dynamic networks with a non-negligible stale entry rate. 


The second part of this section deals with the impact
of multiple buckets per level.
For the evaluation, we first compare KAD to a Kademlia with the same limited number of contacts per level,
i.e. $k_b = 80$ contacts on the first level and $k_i=50$
for $i < b$ on all lower levels, denoted Kademlia80:50.
Furthermore, the 5 buckets on the same level are responsible for different sized fractions of the ID space in KAD.
In order to evaluate the influence of such a skewed split, a KAD (called KAD4) with buckets $111$, $110$, $101$, and $100$
for all lower levels, together with the respective version of Kademlia (Kademlia80:40) is analyzed.
Furthermore, churn and routing table incompleteness are also included in the evaluation, using the
parameters from Section \ref{sec:churn}:
The stale entry rate is chosen as $0.1$, the bucket size is reduced to $0.9$ of its actual value
for the first $10$ levels and $0.8$ for the remaining 	levels.  


Expectations on the hop count can be derived from the expected bit gain per level:
KAD offers a slightly lower bit gain on all levels but the first (See \cite{Stutzbach06improving}, Eq. 6).
In contrast, KAD4 offers a slightly higher expected bit gain than Kademlia80:40 on all levels.
So, we can expect KAD to have a worse performance in terms of the average hop count than Kademlia80:50.
In contrast, KAD4 is bound to outperform Kademlia80:40
from some threshold on, at least in the absence of churn.
Note that the average degree is lower in KAD/KAD4, because the one bucket in the respective Kademlia version is full if all KAD buckets are full, but not vice versa.
Consequently, resilience to node failures should be higher in Kademlia.

The results agree with the above expectations, but also show that the influence of multiple buckets is
small, as can be seen in Figure \ref{fig:resolution}. 
In the absence of churn, the advantage of Kademlia80:50 is barely noticeable, being always in the order of
$0.006$ to $0.007$ hops.
On the other hand, when considering churn and bucket incompleteness, Kademlia80:50 has an advantage of up to $0.04$ hops due to the higher fraction of queries that terminate in the first hops.
The difference between KAD4 and Kademlia80:40 is even less, 
at most $0.002$ hops. 
In the absence of churn, KAD4 indeed achieves a slightly reduced hop count, whereas Kademlia80:40
has a similarly small advantage when considering churn. 
Note that multiple buckets per level reduce the average routing table size by about 7 (KAD) and 2 (KAD4).
Given that routing tables usually contain hundreds to thousands of contacts, such a small
constant storage advantage is negligible.

All in all, our results indicate that multiple buckets reduce the routing table size slightly, but only have a positive effect
on the average hop count if the churn rate is low and the ID space is split equally.
Indeed, the observed advantage of an equal split between multiple buckets can also be achieved by a modified
replacement algorithm in one bucket, which prefers contacts that enhance the diversity of the IDs in the bucket.

\section{Conclusion}
\label{sec:conclusion}

%
%

We have introduced a scalable accurate computation of the hop count distribution in Kademlia-type systems, the only
widely deployed structured P2P systems.
Both simulations and measurements validate our model.
Furthermore, we demonstrated the utility of our model by analyzing common design decisions
in Kademlia-type systems, showing that returning $\beta > 1$ contacts per query is
essential for achieving shorter routes both in static as in dynamic environments.
In addition, we found that having multiple buckets per level
does not necessarily increase the performance but degrades the resilience, and
suggest a modified replacement strategy to combine the higher resilience of
a single bucket per level with the increased bit gain of multiple buckets.

Whereas the model covers all common routing table structure, alterations in the routing
process, such as interleaving queries and recursive routing as well as
a vulnerability analysis of the systems in face of attacks remain future work.

\bibliographystyle{unsrt}
\bibliography{main}

\begin{thebibliography}{10}

\bibitem{Maymounkov02Kademlia}
Petar Maymounkov and David Mazi\`{e}res.
\newblock Kademlia: A peer-to-peer information system based on the xor metric.
\newblock In {\em Proceedings of IPTPS}, 2002.

\bibitem{junemann11towards}
Konrad Junemann et~al.
\newblock Towards a basic dht service: Analyzing network characteristics of a
  widely deployed dht.
\newblock In {\em Proceedings of GridPeer}, 2011.

\bibitem{salah13capturing}
Hani Salah and Thorsten Strufe.
\newblock Capturing connectivity graphs of a large-scale p2p overlay network.
\newblock In {\em Proceedings of HotPOST}, 2013.

\bibitem{Stutzbach06improving}
Daniel Stutzbach and Reza Rejaie.
\newblock Improving lookup performance over a widely-deployed dht.
\newblock In {\em Proceedings of INFOCOM}, 2006.

\bibitem{Steiner10eval}
Moritz Steiner et~al.
\newblock Evaluating and improving the content access in kad.
\newblock {\em Peer-to-Peer Networking and Applications}, 2010.

\bibitem{Jimenez2011subsecond}
Raul Jimenez et~al.
\newblock Sub-second lookups on a large-scale kademlia-based overlay.
\newblock In {\em Proceedings of IEEE P2P Computing}, 2011.

\bibitem{Steiner07global}
Moritz Steiner et~al.
\newblock A global view of kad.
\newblock In {\em Proceedings of IMC}, 2007.

\bibitem{Wang08attacking}
Peng Wang et~al.
\newblock Attacking the kad network.
\newblock In {\em Proceedings of SecureComm}, 2008.

\bibitem{falkner07profiling}
Jarret Falkner et~al.
\newblock Profiling a million user dht.
\newblock In {\em Proceedings of IMC}, 2007.

\bibitem{Crosby07ananalysis}
Scott Crosby and Dan Wallach.
\newblock An analysis of bittorrent’s two kademlia-based dhts.
\newblock Technical report, Rice University, 2007.

\bibitem{Steiner08faster}
Moritz Steiner et~al.
\newblock Faster content access in kad.
\newblock In {\em Proceedings of IEEE P2P Computing}, 2008.

\bibitem{Schiller10gtna}
Benjamin Schiller et~al.
\newblock Gtna: A framework for the graph-theoretic network analysis.
\newblock In {\em Proceedings of Springsim}, 2010.

\bibitem{baumgart09oversim}
Ingmar Baumgart et~al.
\newblock Oversim: A scalable and flexible overlay framework for simulation and
  real network applications.
\newblock In {\em Proceedings of IEEE P2P Computing}, 2009.

\end{thebibliography}

\end{document}